\documentclass[
twocolumn
]{article}
\usepackage{import}

\RequirePackage[
twoside,%
				letterpaper,includeheadfoot,%
				layoutsize={8.125in,10.875in},%
                layouthoffset=0.1875in,%
                layoutvoffset=0.0625in,%
                left=38.5pt,%
                right=43pt,%
                top=43pt,
                bottom=44pt,%
                headheight=0pt,
                headsep=-2pt,%
                footskip=25pt,
                marginparwidth=38pt]{geometry}

\usepackage{microtype}
\usepackage{tablefootnote} 	
\usepackage{hyperref} 		
\usepackage{mfirstuc}       
\usepackage{authblk}

\usepackage{amsmath}
\usepackage{bm} 		
\usepackage{cases}		
\usepackage{amsthm}		
\usepackage{mathrsfs}
\usepackage{siunitx}

\usepackage{graphicx}
\usepackage{caption, subcaption}

\usepackage[
    backend=biber,
    style=nature,
  ]{biblatex}

\usepackage{xparse}

\def\eqprefix{Eq.~}
\def\figprefix{Fig.~}

\def\thmprefix{Thm.~}
\def\SIprefix{\S}
\newcommand\subfigname[1]{.#1}

\newtheorem{theorem}{Theorem}
\newtheorem{lemma}{Lemma}

\newcommand{\hg}{karigami}
\newcommand{\cell}{scissor}

\newcommand{\IC}[0]{intrinsic compatibility}
\newcommand{\KC}[0]{kinematic compatibility}
\newcommand{\SA}[0]{collapsibility}

\newcommand{\cij}[1]{\ensuremath{\text{C}_{\left(#1\right)}}}

\NewDocumentCommand\pivotrotationmap{mm}{
	f_{#2}^{(#1)}
}

\NewDocumentCommand\cellsymbol{m}{
	\IfNoValueTF{#1}
	{\ensuremath{\bm{S}}}
	{\ensuremath{\bm{S}_{#1}}}
}

\NewDocumentCommand\leg{m o}{
	\IfNoValueTF{#2}
	{\ensuremath{\bm{v}^{(#1)}}}
	{\ensuremath{\bm{v}^{(#1)}_{#2}}}
}

\NewDocumentCommand\vertex{m}{
\ensuremath{\bm{x}_{#1}}
}

\NewDocumentCommand\pivot{o}{
\ensuremath{\bm{p}_{#1}}
}

\NewDocumentCommand\leglength{m o}{
	\IfNoValueTF{#2}
	{\ensuremath{\ell^{(#1)}}}
	{\ensuremath{\ell^{(#1)}_{#2}}}
}

\NewDocumentCommand\boundary{o}{
	\IfNoValueTF{#1}
	{\ensuremath{t}}
	{\ensuremath{t_{#1}}}
}
\NewDocumentCommand\sboundary{o}{
	\IfNoValueTF{#1}
	{\ensuremath{s}}
	{\ensuremath{s_{#1}}}
}

\title{Additive design of 2-dimensional scissor lattices}

\author[1]{Noah Toyonaga}
\author[1,2,3]{L. Mahadevan}

\affil[1]{Department of Physics, Harvard University}
\affil[2]{Department of Organismic and Evolutionary Biology, Harvard University}
\affil[3]{School of Engineering and Applied Sciences, Harvard University}

\begin{document}

\twocolumn[
  \begin{@twocolumnfalse}
  \maketitle
    \begin{abstract}
We introduce an additive approach for the design of a class of transformable structures based on two-bar linkages (``scissor mechanisms'') joined at vertices to form a two dimensional lattice. 
Our discussion traces an underlying mathematical similarity between linkage mechanisms, origami, and kirigami and inspires our name for these structures: karigami. 
We show how to design karigami which unfold from a one dimensional collapsed state to two-dimensional surfaces of single and double curvature. 
Our algorithm for growing \hg{} structures is provably complete in providing the ability to explore the full space of possible mechanisms, and we use it to computationally design and physically realize a series of examples of varying complexity.
\end{abstract}
  \end{@twocolumnfalse}
]

Inspired by art \cite{Kasahara1999-fz, Lang2011-zm}, enabled by mathematics \cite{Callens2018-lv, Lang1996-td, Mahadevan2005-hv}, and egged on by the allure of applications \cite{Treml2018-kt, Rus2018-fm, Melancon2021-az, Bertoldi2017-xf}, scientists and engineers have investigated the principles which govern the paper arts of origami and kirigami and deployed these insights to design metamaterials with programmable shape and functionality \cite{Choi2019-nr, Dudte2021-hq, Dudte2023-ik}. It is useful to reiterate the basic ideas behind origami and kirigami which are distinguished by their fundamental kinematic mechanisms---the fold in origami, and the cut/pivot in kirigami\footnote{Indeed, the etymology of these art forms belies the science: the Japanese verb ``oru'' refers to folding while ``kiru'' refers to cutting. 
Taking literally the ``scissoring'' motion underlying the present mechanisms we propose the name Karigami from ``karu'', which refers to shearing.}
In the present work we move beyond cutting and folding and ask a related question: what rules govern two-dimensional structures constructed using a two bar linkage as the unit cell? 

The two bar linkage is one of the essential building blocks of complex mechanisms while also recognizable in the quotidian form of a pair of scissors.
Assemblies built from scissor-like subunits appear across natural and engineered systems:
from nanometer-scale contractile injection machinery of bacteriocidal viruses \cite{Toyonaga2024-ma},
to handheld pantographic drawing aids and architectural gridshells \cite{Sigrid2014-ij}.
Here we distill the general geometric principles which underlie these mechanisms assembled from lattices of scissors to develop an algorithm for the design of surfaces with programmable shape and functionality, drawing on two related intellectual genealogies: the engineering of kinematic linkages 
\cite{Kempe1875-be, McCarthy2010-qr, Moon2009-ql,Maden2011-nn, Chen2005-qo, Demaine2007-mg} and the design of architectural gridshells \cite{Sigrid2014-ij, Quinn2014-sy, Baek2018-pn, Sageman-Furnas2019-kf, Panetta2019-hw, Becker2023-lv}.

\begin{figure*}[hbt!]
\begin{center}
	\includegraphics[width=7in]{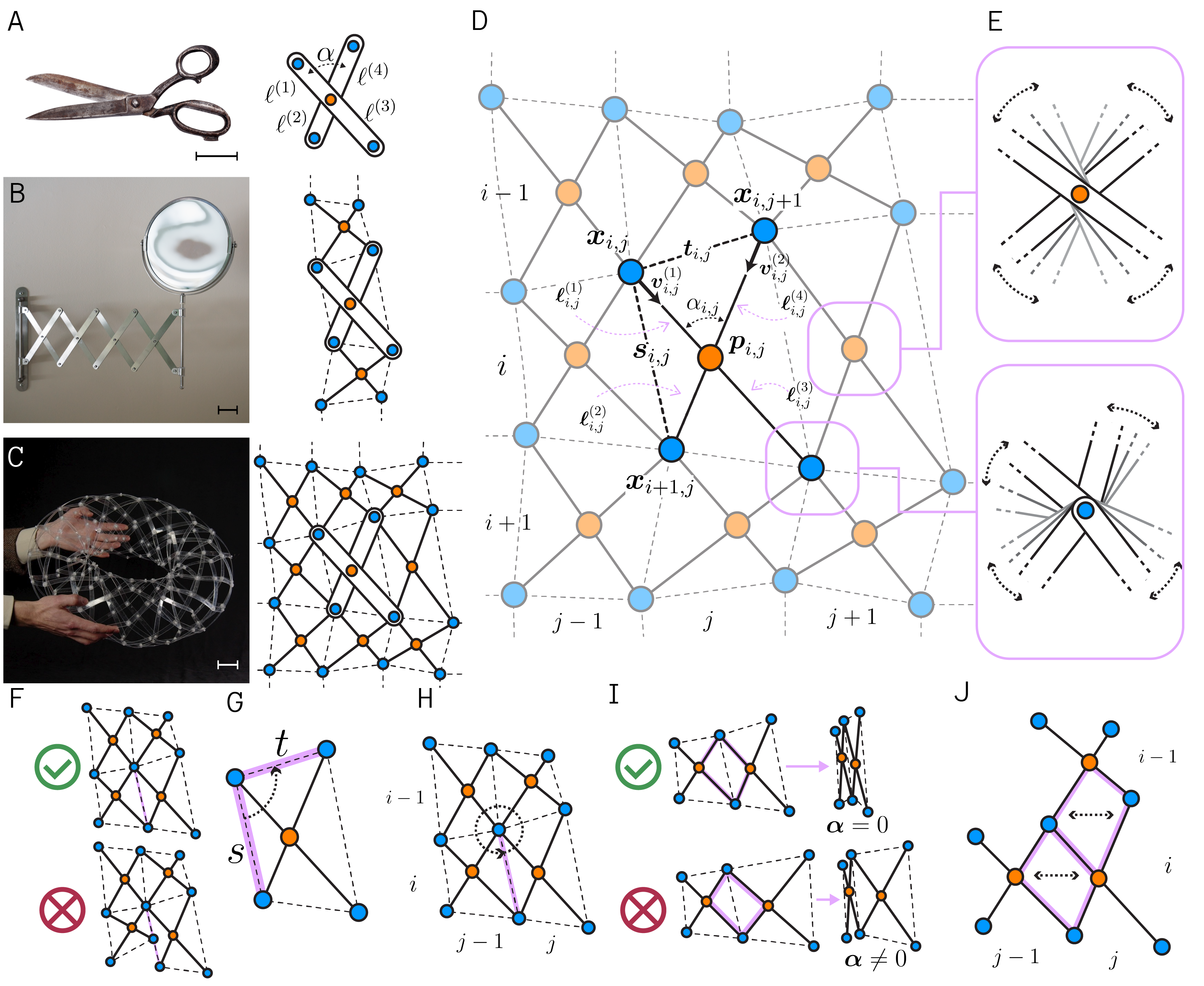}
\end{center}
\caption{
    \textbf{Geometry and Kinematics}
    (A-C) Physical examples of assemblies made from scissors arranged in 0-, 1-, and 2-D lattices. Scale bars represent 5cm.
    (A) The geometry of a single scissor is characterized by four leg lengths and an opening angle. 
    (B) The scissor mount for a bathroom mirror is an example of a 1-D chain of identical scissor mechanisms. 
    The coupling of opening angles allows for the telescopic extension of the mirror from the wall.
    (C) A \hg{} torus designed using the methods developed in this paper is characterised by a 2-D lattice of scissor elements.
	(D) \hg{} consists of a lattice of scissor-like unit cells (\cellsymbol{}) composed of two rods (represented by black lines) joined by a \textit{scissor pivot} (orange circle). 
	Multiple cells can be joined together using \textit{vertex pivots} (blue).
	Throughout this paper, we represent the geometry of each \cell{} by the combination of leg direction vectors ($\{\leg{1}, \leg{2}\}$) and leg lengths ($\{\leglength{1}\ldots \leglength{4}\}$).
	In addition we introduce the \textit{edge vectors}, $\bm{s}$ and $\bm{t}$ which lie along the upper and left sides of the \cell{} facet respectively.
	(E) scissor and vertex pivots are distinguished by the way they connect strips.
	(F) For a set of cells to be intrinsically compatible, the lengths of adjacent facets must be equal (see edges highlighted in pink).
	(G-H) We can derive a closure condition for intrinsic compatibility by considering mapping facet edge lengths around the vertex \vertex{i,j}.
	(I) \SA{} describes whether a \hg{} structure can be ``flattened'', in the sense that there exists an intrinsically valid state where all opening angles $\{\alpha\}$ are identically 0. Note that in the lower example one cell has $\alpha=0$ while in the other $\alpha>0$.
	(J) \SA{} (\eqprefix{}\ref{eq:collapsibility}) requires that the total length of the legs on the left of each pore must equal the total length of the legs on the right.
}
\label{fig:1_GEOMETRY_AND_KINEMATICS}
\end{figure*}


By considering a simple lattice of scissor based mechanisms, we find a set of compatibility conditions which are analogous to the constraints which govern the design of origami and kirigami \cite{Dudte2021-hq, Dudte2023-ik}, inspiring our name for these structures: karigami. 
We use these conditions to develop an additive design procedure that allows us to explore of the full space of collapsible structures. 
We finally present a sampling menu of possible structures and realize these as physical models. 

The atomic unit of a \hg{} structure is a two bar linkage, a \textit{\cell{}}, as illustrated in \figprefix\ref{fig:1_GEOMETRY_AND_KINEMATICS}\subfigname{A}.
Each \cell{} (denoted \cellsymbol{}) has 5 degrees of freedom, the 4 leg lengths $\leglength{a}$ (measured relative to the \cell{} pivot) and the opening angle $\alpha$ between them.
In practice, we will treat the leg lengths as design parameters, while retaining the opening angle $\alpha$ as a kinematic degree of freedom.
For convenience, we also introduce the notion of the \cell{} \textit{facet} as the quadrilateral bounded by the tips of the legs (note that the shape of the facet changes as a function of the deployment angle $\alpha$).

A set of scissors ($\{\cellsymbol{i}\}$) can be joined together to form a 1-D chain via a secondary set of joints we call \textit{vertex pivots} which connect pairs of legs at their tips; a common example is a bathroom mirror that can extend from the wall, shown in \figprefix{}\ref{fig:1_GEOMETRY_AND_KINEMATICS}\subfigname{B} which has both \textit{scissor pivots} (orange) and \textit{vertex pivots} (blue).
The introduction of vertex pivots couples the opening angles of every pair of adjacent scissors, and thus the overall geometry of the structure must be parameterized by a single kinematic variable $\theta$ (e.g. the opening angle of the first pair of scissors in the chain).

Extending this framework to a two dimensional cartesian lattice of scissors $\{\cellsymbol{i,j}\}$ yields a \textit{karigami} lattice (\figprefix\ref{fig:1_GEOMETRY_AND_KINEMATICS}\subfigname{C}), where again
 scissor pivots  join rods within a scissor and vertex pivots connect rods of multiple scissors. Note that vertex pivots constrain the position of connected rods but not their orientation; in 3D this flexibility permits vertices in karigami lattices to develop angle defects which in turn allows us to introduce curvature and control geometry. 
\footnote{From an engineering perspective both scissor and vertex pivots are properly revolute joints in the 2D setting, while in 3D the vertex pivots must be spherical joints to accommodate out-of plane displacements.
As we demonstrate with in physical models, however, revolute connections can be used for all joints -- even in 3D structures -- provided the rods are sufficiently flexible.}

As in the case of the one dimensional lattice, 
the opening angles of all \cell{}s are coupled so the kinematic configuration $\mathbf{C}\equiv \left\{ \alpha_{i,j} \right\}$ of the structure can be parameterized by a single number $\theta$:
$\mathbf{C} = \mathbf{C}(\theta) = \left\{\alpha_{i,j}(\theta)\right\}$.
Unlike the one-dimensional lattice, however, two dimensional lattices are generically rigid, and operation as a continuous mechanism requires a \KC{} condition which we will discuss in detail below.

We begin by analyzing the simplest nontrivial \hg{} lattice consisting of four \cell{}s joined with a single vertex pivot as shown in \figprefix{}\ref{fig:1_GEOMETRY_AND_KINEMATICS}\subfigname{F}.
An arbitrary set of four \cell{}s can not generally be joined into a valid \hg{} assembly due to the coupling of opening angles between cells, as illustrated in the lower sub-figure of \figprefix{}\ref{fig:1_GEOMETRY_AND_KINEMATICS}\subfigname{F}.
We say that a set of cells is \textit{intrinsically compatible} if, for some configuration, the lengths of the edges of abutting facets are equal.
\footnote{Of course, it is straightforward to design lattices of \cell{}s which are intrinsically compatible: start from a mesh of quads and draw the appropriate diagonals to construct the pivot center and legs.
However, such structures will be generically un-deployable as we discuss below.}

To derive conditions for the intrinsic compatibility we write down a closure condition involving the lengths of legs of the cells joined around a central vertex. For convenience we introduce the \textit{leg vectors} (
$\leg{1}[i,j]\equiv \bm{p}_{i,j} - \vertex{i,j} / \|\bm{p}_{i,j} - \vertex{i,j}\|$, 
$\leg{2}[i,j]\equiv \bm{p}_{i,j} - \vertex{i,j+1} / \|\bm{p}_{i,j} - \vertex{i,j+1}\|$~) which,
along with the leg lengths (\leglength{1}[i,j], \leglength{2}[i,j], \leglength{3}[i,j], \mbox{\leglength{4}[i,j]),} completely determine the geometry of each cell. In \figprefix\ref{fig:1_GEOMETRY_AND_KINEMATICS}\subfigname{E}, we show the relationship between the lengths of the left and top edges of a \cell{} ($s$ and $t$ respectively). 
If we assume $0<\alpha<\pi/2$ there is a linear map $f$ between $s^2$ and $t^2$ given by \eqprefix\ref{eq:cell_edge_relationship}, where we have introduced the symbols 
$\gamma^{(a,b,c)}\equiv {\leglength{b}}^2 + {\leglength{c}}^2 + \frac{\leglength{c}}{\leglength{a}}\left({\leglength{a}}^2 + {\leglength{b}}^2\right)$ 
and $\eta^{(a, c)} \equiv \frac{\leglength{c}}{\leglength{a}}$.

\begin{equation}
	\label{eq:cell_edge_relationship}
	t^2 = f_{i,j}^{(2,1,4)}\left( s^2 \right)
	\equiv \gamma^{(2, 1, 4)} + \eta^{(2,4)}s^2
\end{equation}

We now turn to the general case of a set of four \cell{}s arranged around a shared vertex, as illustrated in \figprefix\ref{fig:1_GEOMETRY_AND_KINEMATICS}\subfigname{H}. 
By recursively applying the relationship found between the adjacent edges of a single \cell{} (\eqprefix\ref{eq:cell_edge_relationship}) we can define a map $g$ that gives $s'^2$, the squared length of the boundary between \cell{}s \cij{i,j-1} and \cij{i,j}, given by traversing a loop around the vertex $\vertex{i,j}$:

\begin{equation}
	\label{eq:intrinsic_compatibility}
	\begin{split}
	g_{i,j}(s^2) 
	& \equiv
	\pivotrotationmap{1,4,3}{i,j-1} \circ
	\pivotrotationmap{4,3,2}{i-1,j-1} \circ
	\pivotrotationmap{3,2,1}{i-1,j} \circ
	\pivotrotationmap{2,1,4}{i,j} 
	\left( s^2 \right)
\end{split}
\end{equation}

A four-cell mechanism is geometrically overconstrained, and additional symmetries associated with \textit{\KC{}} are required to introduce internal degrees of freedom. Then, if for a continuous range of the kinematic parameter $\theta_0<\theta<\theta_1$ the cells are intrinsically compatible the following lemma must hold: 

\begin{lemma}
	If a four \cell{} mechanism is intrinsically valid in two kinematic states $\mathcal{S}(\theta_0)$, $\mathcal{S}(\theta_1)$ the mechanism admits a one parameter family of continuous deformations between them.
\end{lemma}

\begin{proof}

Four arbitrary \cell{}s are intrinsically compatible if, for some kinematic state (parameterized by $\theta$), $s(\theta)^2 = g(s(\theta)^2) $ holds.

We note that as $g$ is a composition of linear functions and therefore is itself linear (\eqprefix{}\ref{eq:intrinsic_compatibility}).
Then, if are two distinct values ($s^2$) of under which $s^2 = g(s^2)$ holds it must be the case that $g$ is the identity map (i.e., $g(s^2) = s^2$).
Thus \IC{} holds for all states $\theta_0<\theta<\theta_1$.
\end{proof}

\begin{figure*}[hbt!]
\begin{center}
	\includegraphics[width=\textwidth]{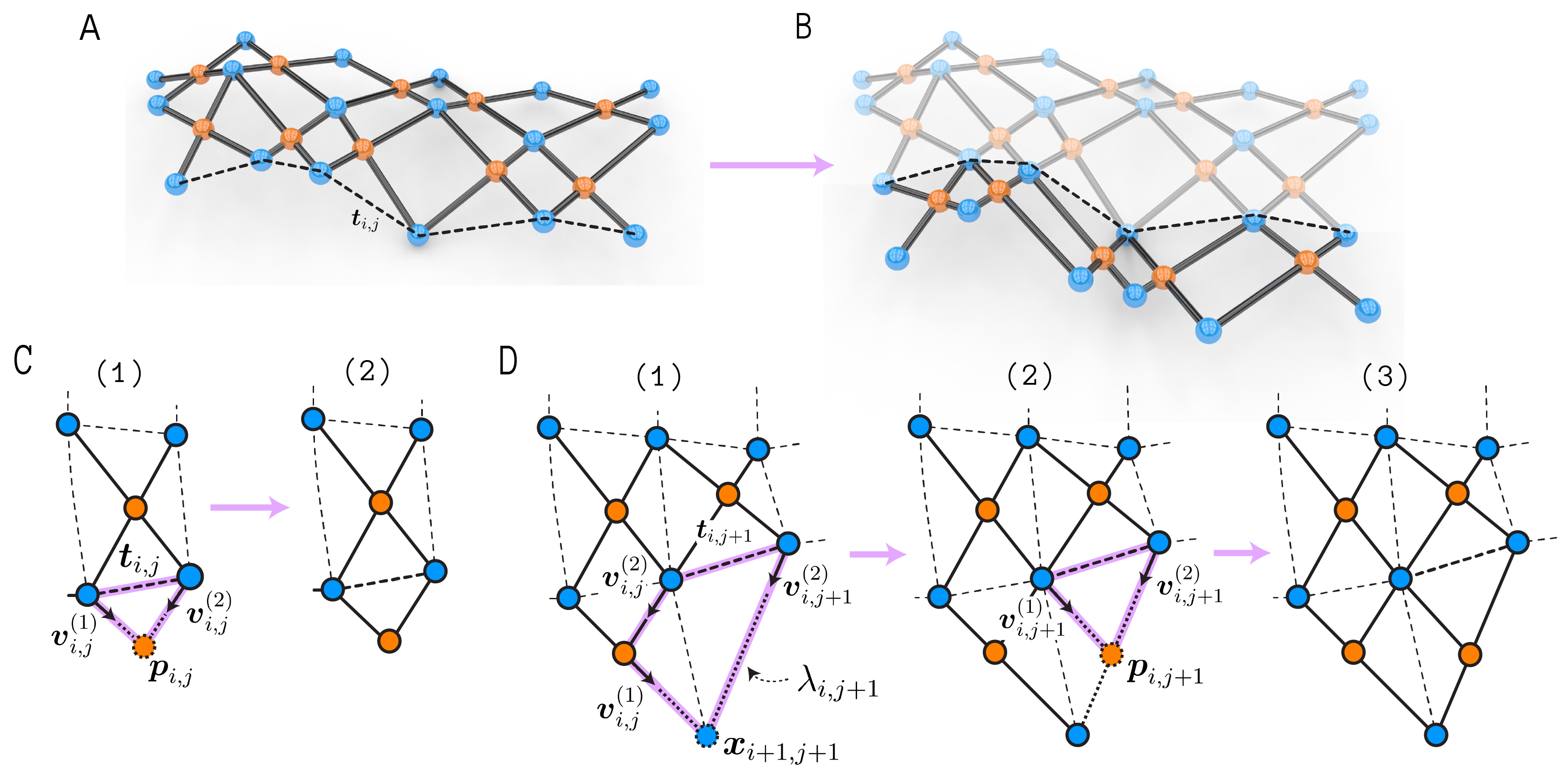}
\end{center}
\caption{
    \textbf{Constraints and Construction}
    (A) An existing karigami lattice, with the growth front ($\{\boundary[i,1], \ldots , \boundary[i,n]\}$) identified using a dotted line.
    (B) The additive algorithm constructs a new row of \cell{}s along the growth front.
	(C) Construction of the scissor pivot of a new \cell{} at site $j$ along the growth front. 
    (C.1) We are free to a leg direction \leg{2}[i,j].
    The triangular loop highlighted in (C.1) then corresponds to the vector closure condition \eqprefix{}\ref{eq:2D_pivot_closure}. 
    (C.2) \eqprefix{}\ref{eq:1D_leg_length} sets \leglength{4}[i,j], and thereby fixes the location of pivot \pivot[i,j] (as well as \leg{1}[i,j], \leglength{1}[i,j]).
	(D) Construction of an scissor pivot \pivot[i,j+1] adjacent to a \cell{} (\cellsymbol{i,j}) with an existing scissor pivot.
    (D.1) First, we construct the unique vertex pivot \vertex{i+1,j+1} making use of the vector closure condition highlighted in pink.
    (D.2) Second, we construct \pivot[i,j+1] according to the same process as shown in (C), with the direction \leg{2}[i,j+1] set by the location of the (now fixed) vertex \vertex{i,j+1}.
    (D.3) The new row of karigami with \pivot[i,j] and \pivot[1,j+1] constructed.
}
\label{fig:2_ADDITIVE}
\end{figure*}

In practical terms, if a four-cell mechanism with fixed leg lengths has two distinct deployed states which are intrinsically compatible, it can be kinematically deformed between them.
Note that we have made no assumption about the relative orientations of the facets about the central vertex \vertex{i,j}; indeed, facets are generically not parallel to one another, generating a surplus or deficit of angular material localized at the vertex as previously discussed.

For a larger structures with more than one bulk vertex pivot (e.g. \figprefix\ref{fig:1_GEOMETRY_AND_KINEMATICS}\subfigname{B}), similar \KC{} conditions must hold about every loop of \cell{}s surrounding one or more vertices. If there are two distinct deployed states ($\mathbf{C}(\theta_1)$, $\mathbf{C}(\theta_2)$) for the entire lattice, \eqprefix{}\ref{eq:intrinsic_compatibility} will be satisfied  for the continuum of states $\{\mathbf{C}(\theta)| \theta_1<\theta<\theta_2\}$.

\textbf{Collapsibility}
In addition to \KC{}, we introduce a geometric property which we term \SA{}. Analogous to the notion of developability in origami~\cite{Dudte2021-hq} or contractibility in kirigami~\cite{Dudte2023-ik}, \SA{} 
describes whether a structure can be completely ``flattened'' with every \cell{} in either a closed state as illustrated in \figprefix\ref{fig:1_GEOMETRY_AND_KINEMATICS}\subfigname{I}. 
To derive the conditions for collapsibility consider the set of three adjacent \cell{}s $\{ \cij{i,j-1}, \cij{i-1,j}, \cij{i,j} \}$, as illustrated in \figprefix\ref{fig:1_GEOMETRY_AND_KINEMATICS}\subfigname{J}.
In order for all three \cell{}s to fold to a state $\alpha\rightarrow 0$ (close-sheared) the sum of the length of legs on the left of each quad must equal the sum of the lengths of the legs on the right 
\footnote{We note that an alternate collapsed state may exist in which the cells are fully ``opened'' in the sense that for some $\theta=\theta_o$, we have $\bm\alpha(\theta_o) = \pi/2$ (\eqprefix{}S.1 gives the formulas analagous to \eqprefix{}\ref{eq:collapsibility} for such ``open'' collapsed states).
As the conditions for 0- and $\pi/2$- collapsibility are equivalent up to a change of coordinates $\alpha \rightarrow \pi/2 - \alpha$ and the only \hg{} lattices which have two collaped states can be trivially mapped to chebyshev nets (see discussion in \SIprefix{}S.I), we will simply refer to collapsibility in the sense defined above.
}. 

\begin{subequations}
	\begin{align}
		\leglength{3}[i,j-1] + \leglength{4}[ i,j-1 ] = \leglength{1}[ i,j ] + \leglength{2}[ i,j ] & \quad \text{(horizonal)} \label{eq:collapsibility_horizontal} \\
		\leglength{2}[i-1,j] + \leglength{1}[i,j] = \leglength{3}[ i-1,j ] + \leglength{4}[ i,j ] & \quad \text{(vertical)} \label{eq:collapsibility_vertical}
\end{align}\label{eq:collapsibility}
\end{subequations}

\eqprefix{}\ref{eq:collapsibility} generalizes the collapsibility conditions for a one dimensional chain of scissors found by \cite{Escrig1993-uk}.


\begin{figure*}[hbt!]
\begin{center}
	\includegraphics[width=\textwidth]{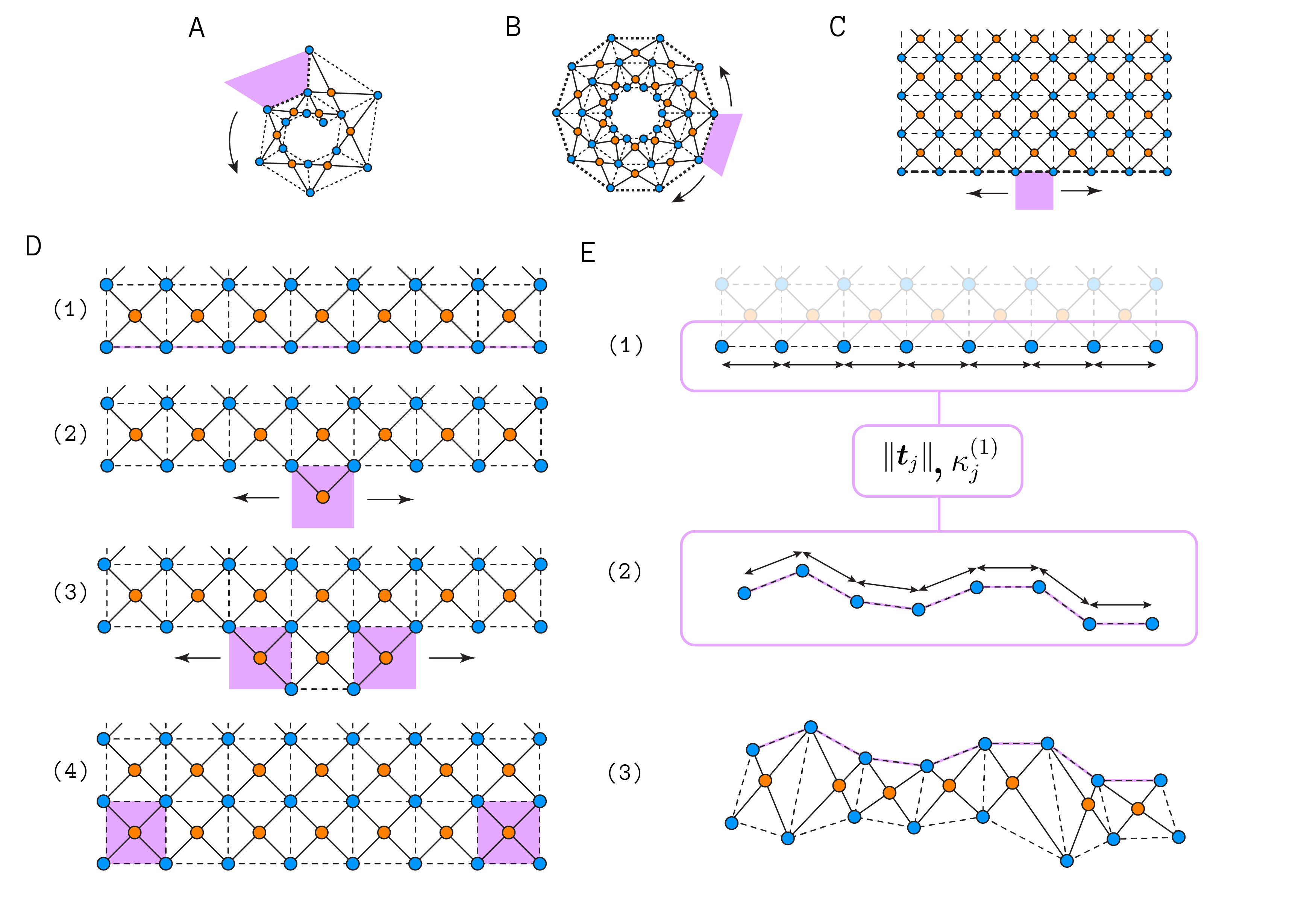}
\end{center}
\caption{
	\textbf{Growth of Surfaces}
    (A-C) \hg{} surfaces can be categorized by the topology of their constitutive growth fronts (equivalently, by the topology of the seed with which these surfaces were initialized). 
    Thus we distinguish between helices (A) which are fully determined by their initial seed, loops (B) with one free edge, and sheet-like surfaces (C) with two free edges. 
    (D) In the \textit{extrinsic algorithm} we consider constructing a new row of \cell{}s along the edge of an existing \hg{} structure---in this case, the last row of the structure shown in (C). 
    The algorithm is initialized by first constructing a pivot at a site $j$ (D.2). 
    Using our algorithm we can construct the \cell{}s left and right along the front (D.3) until a full row of \cell{}s have been built (D.4).
    (B-D) \hg{} can be classified by the topologies of the ``seed'' row by which they are initialized.
    In \figprefix{}\ref{fig:4_EXAMPLES} we demonstrate structures designed with each of these seed topologies.
    E) The \textit{intrinsic} algorithm generalizes the growth procedure presented in (D) by considering not only the growth front given by the space curve $\{t_j\}$, but a class of compatible fronts.
    (E.1) We begin with an existing karigami growth front and extract the sequence of edge lengths 
    $\{|t_j|\}$ and shearing $\{k_j\}$ parameters for each site.
    (E.2) A compatible growth front is given by any space curve with the same sequence of edge lengths extracted from the karigami front as in (E.1).
    (E.3) Finally, we apply the extrinsic algorithm to the modified front to construct a strip of cells  that is \IC{} and \KC{} with the original karigami structure.
}
\label{fig:3_ALGORITHM}
\end{figure*}

\textbf{Additive Construction}
The notions of \KC{} and \SA{} allow us to pose the question of designing \hg{} lattices as an additive construction problem by asking how to attach a new \cell{} to an existing \hg{} lattice (shown schematically in \figprefix\ref{fig:2_ADDITIVE}) such that the resulting structure remains both kinematically compatible, and collapsible. 

This corresponds to an iterative algorithm that additively builds an karigami lattice by adding scissors that satisfy \eqprefix{}\ref{eq:intrinsic_compatibility}, \ref{eq:collapsibility}.
This allows us for us to design the collapsed and deployed states simultaneously, ensuring that there is a continuous kinematic path between them.
We note that this deployment path generically passes through states of geometric incompatibility, analogous to those present when folding or unfolding origami \cite{Dudte2016-va, Choi2019-nr}.


To characterize the geometry of adding scissor pivots along the growth front, we first note the following theorem: 

\begin{theorem}\label{theorem:additive_algorithm}
	
The space of scissor pivots along the growth front of a \hg{} lattice in $D$ dimensions has $(D-1)$ DOF. 

\end{theorem}

\begin{proof}
	Our proof can be broken down into three steps: 
	first, conditions for construction of a single scissor pivot on a growth front of a \hg{} lattice;
	second, construction of adjacent scissor pivots (i.e. the scissor pivot of adjacent \cell{}s); 
	and finally construction of all scissor pivots along a growth front.

	\textit{Single \cell{} Construction}
Consider adding a \cell{} at site $j$ (i.e. with upper boundary given by $\boundary[i,j]$) to the edge of a \hg{} lattice as illustrated in \figprefix\ref{fig:2_ADDITIVE}\subfigname{A}. 
	To begin the construction we first specify the direction of the leg direction \leg{2}[i] (e.g. in 2D by a rotation angle $\phi\in\left(0,\pi\right)$ about \vertex{i,j+1} relative to $\bm{t}$) as shown in \figprefix\ref{fig:2_ADDITIVE}\subfigname{C.1}.
	Now, since the lengths $\leglength{1}[i,j]$, $\leglength{4}[i,j]$ and the remaining leg direction $\leg{0}[i,j]$ are coupled by the vector closure condition \eqprefix\ref{eq:2D_pivot_closure}, 
	all three are determined once any one is known (see \figprefix{}\ref{fig:2_ADDITIVE}\subfigname{C.2}).

	\begin{equation}
		\leglength{1}[i,j]\leg{1}[i,j] = \boundary[i,j] + \leglength{4}[i,j]\leg{2}[i,j]
	\label{eq:2D_pivot_closure}
	\end{equation}

	We can solve \eqprefix{}\ref{eq:2D_pivot_closure} together with \eqprefix\ref{eq:collapsibility_vertical} to find the unique leg length \leglength{4}[i] as given by \eqprefix{}\ref{eq:1D_leg_length}. 
	Here we introduce a constant $\kappa_{i,j}\equiv \leglength{2}[i-1,j]-\leglength{3}[i-1,j]$ we term the \mbox{\textit{collapsibility parameter}} which measures the asymmetry between the legs of the preceding cell on the growth front. 

	\begin{equation}
		\leglength{4}[i,j] = \frac{{\kappa_{i,j}}^2 - ||\bm{t}_{i,j}||^2}{2 \left(\kappa_{i,j}+ \bm{t}_{i,j}\cdot \leg{2}[i,j]\right)}
		\label{eq:1D_leg_length}
	\end{equation}
	
	Setting \leglength{4}[i] according to \eqprefix{}\ref{eq:1D_leg_length} fixes the location of the \cell{} pivot \pivot[i,j], as illustrated in \figprefix{}\ref{fig:2_ADDITIVE}\subfigname{C.2}.
	(Note that one could have equivalently proceeded by setting \leg{1}, and solving \eqprefix{}\ref{eq:2D_pivot_closure} and \eqprefix\ref{eq:collapsibility_vertical} for \leglength{4}. 
    In either case, we find the unique scissor pivot \pivot[i,j] which satisfies \SA{} with the \cell{} above, \cellsymbol{i-1,j}.)
	For derivation of \eqprefix{}\ref{eq:1D_leg_length} and further discussion see \SIprefix{}S.II.

	If $m=1$ this is a chain and the \cell{} has no neighbors there are no further pivots to construct, and the only freedom we have is to set the direction of a single leg vector, \leg{2}[0], i.e. $D-1$ DOF. In the case $m>1$, we will show the location of adjacent pivots (\pivot[i,j], \pivot[i,j+1]) are coupled by the \SA{} equations (\eqprefix{}\ref{eq:collapsibility}) associated with the corresponding \cell{}s (\cellsymbol{i,j}, \cellsymbol{i,j+1}).
	
	\textit{Construction of adjacent \cell{}s}
    We now show the existence of one scissor pivot \pivot[i,j] along the growth front gives rise to a unique compatible neighbor, \pivot[i,j +1]. Assuming that the scissor pivot \pivot[i,j] of the preceding \cell{} \cellsymbol{i,j} is known,  the location of vertex pivot \vertex{i+1,j+1} is uniquely determined (see \SIprefix{}S.II and \figprefix\ref{fig:2_ADDITIVE}\subfigname{D.1}).  


    The new pivot fixes the leg direction \leg{2}[i,j+1], of this scissor (see \figprefix{}\ref{fig:2_ADDITIVE}\subfigname{D.2}) and we can thus immediately use the preceding result for single pivot construction to establish the location of \pivot[i,j+1]. 
    (Details of this derivation provided in \SIprefix{}S.II.)  Iterating this procedure along the growth front allows us to construct a string of scissor pivots one site at a time.

\begin{figure*}[htb!]
\begin{center}
	\includegraphics[width=\textwidth]{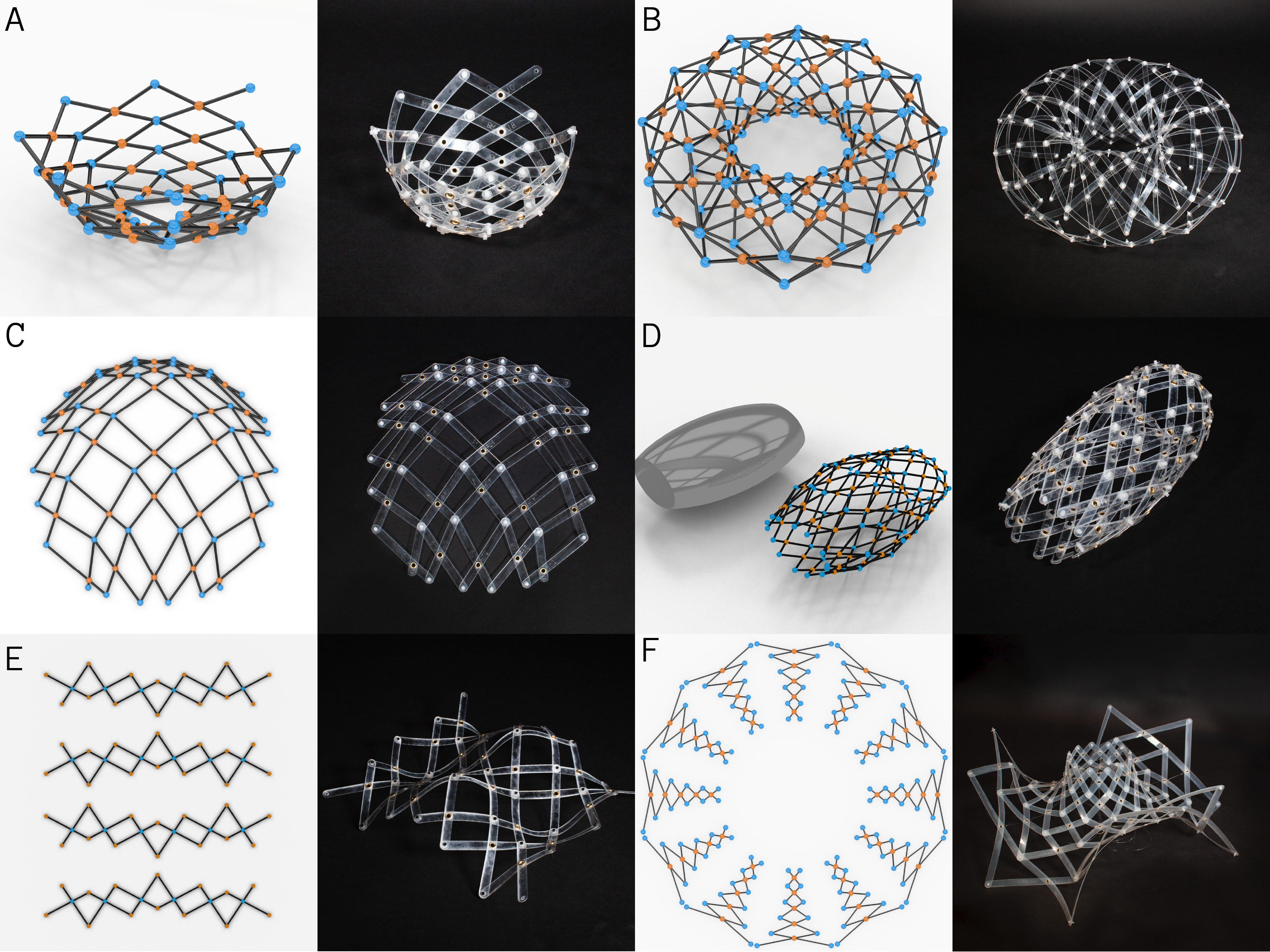}
\end{center}
\caption{
	Karigami structures designed using our extrinsic (examples A-D) and intrinisic (examples E,F) algorithms. 
    (See \SIprefix{}S.IV for details of surface construction.)
    (A) A bowl constructed from a helix-like seed with 32 steps of growth.
    (B) An axisymmetric torus. 
    Here the initial seed is a loop lying along a nodal curve along the top (or bottom) of the torus.
    (C, D) We design a lattices which map the surface of a disk and an ellipsoid. 
    Here we optimize over the growth parameters for each row (i.e. initial pivot locations and terminal leg lengths) to fit the target shapes. 
    (E) Using the intrinsic algorithm we develop an eggbox-like surface with multi-stable "nodes" which can be flipped in the deployed state (and reset by closing the karigami).
    (F) We also use the intrinsic algorithm to design a surface of constant negative Gaussian curvature.
}
\label{fig:4_EXAMPLES}
\end{figure*}

 	\textit{Growth along entire front.}

    The preceding observation can be formalized by recognizing that our geometric construction implies a bijective and invertible transfer function $g_{i}$ which maps the pivot locations between adjacent \cell{}s: 

	\begin{equation}
		\pivot[j+1] = g_{j+1}\left( \pivot[j] \right)
	\end{equation}

	Now consider constructing \pivot[ i,j' ] given \pivot[i,j] assuming $j'>j$. 
	We can do so using a composition of the transfer function 

    \begin{equation}
    \begin{split}
        \pivot[i,j'] &= g_{i,j'-1}(g_{i,j'-2}(
		\cdots g_{i,j+1}(g_{i,j}(\pivot[i,j])) 
		\cdots )) \\
        &= f(\pivot[i,j])
    \end{split}
	\end{equation}

	As all $g$ are bijective their composition is bijective and thus all leg directions $\left\{\leg{2}\right\}$ along the front are set by a single leg direction \leg{2}[i,j]. 
    An equivalent argument holds for $j'<j$ using a composition of inverse transfer functions $g^{-1}i$.
    Thus, for any front with m scissors ($m\geq 1$), the space of all possible new scissor pivot locations is set by any single leg direction, which has $D-1$ degrees of freedom.

\end{proof}

Given an existing \hg{} lattice, \thmprefix{}\ref{theorem:additive_algorithm} guarantees we can construct a new rows of \cell{}s (\figprefix{}\ref{fig:3_ALGORITHM}\subfigname{D}) such that the complete structure satisfies \KC{} and \SA{}. 
This suggests a natural procedure to design \hg{} by iterative assembling each successive row of \cell{}s.
We distinguish between and describe two versions of this procedure: the Extrinsic Algorithm develops the geometry of flat-facet karigami, 
while the Intrinsic Algorithm enables the construction of bent-facet karigami.
We note that in either case, the construction procedure must be constrained to the set of valid leg directions at each site (see discussion in derivation of \thmprefix{}\ref{theorem:additive_algorithm} and \SIprefix{}S.II).

To initiate construction using either algorithm, we must first specify a minimal seed consisting of (1) an initial set of vertices which form the upper boundary of a row of \cell{}s and (2) a set of $\kappa$ values that set the leg length differences of this initial row.
We identify three classes of karigami seeds  (see \figprefix{}\ref{fig:3_ALGORITHM}\subfigname{A-C}), which are distinguished by the geometry and topology of the resulting structures: helical, tubular, and sheet karigami (See \SIprefix{}S.III for further discussion.)

\textbf{Extrinsic Algorithm: Flat-Facet karigami}

To design flat-facet karigami (\figprefix{}\ref{fig:3_ALGORITHM}\subfigname{D}), we assign a growth front to an existing karigami structure (or choose an karigami seed, as described above), and construct successive rows of \cell{}. 
Explicitly, one proceeds as follows: 

\begin{itemize}
	\item Starting at any site $j$ along the front, assign the direction \leg{2}[i,j].
	\item Propagate \leg{2}[i,j] to all other \cell{}s $\left(\cellsymbol{ i,0 }, \cdots{} \cellsymbol{ i,m }\right)$ along the front by iteratively applying the geometric construction procedure previously discussed.
	\item At boundaries, assign the initial and final leg lengths (\leg{2}[i,0], \leg{3}[i,m]).
	\item Identify a new growth front along any boundary and repeat this procedure to grow a new strip.
\end{itemize}

\textbf{Intrinsic Algorithm: Bent-Facet karigami}
To design bent-facet karigami (\figprefix{}\ref{fig:3_ALGORITHM}\subfigname{E}) we introduce an additional step in which we modify the geometry of the growth front before the construction procedure. 


Note that while the lengths between vertices on the growth front and corresponding shearing parameters are fixed by the geometry of the preceding row (which, taken together, guarantee \KC{} and \SA{}), we are free to pick the orientation of each facet vector along the growth front before we construct a new row of \cell{}s.
Thus, our construction proceeds as follows:

\begin{itemize}
	\item Extract the shearing parameters $\{\kappa_{j}\}$ and lengths of facets $\{\|\boundary{j}\|\}$ along the growth front.(\figprefix{}\ref{fig:3_ALGORITHM}\subfigname{E.1})
	\item Construct a new compatible growth front (i.e. a front with equivalent $\{\kappa_{j}\}$, $\{\|\boundary{j}\|\}$) (\figprefix{}\ref{fig:3_ALGORITHM}\subfigname{E.2})
	\item Apply the extrinsic algorithm to the new front (\figprefix{}\ref{fig:3_ALGORITHM}\subfigname{E(3)})
\end{itemize}


The physically joining the scissor strips designed using this procedure generally involves geometric incompatibilities and bending of \cell{} facets.
As such the deployed geometry of a structure designed using Algorithm II will depend on the material properties used for their physical realization (see \figprefix{}S.7).

We can extend this procedure for programming intrinsic geometry to entire structures: at each stage of growth we first design a compatible seed $\{\boundary[i,j]^*\}$ which preserves the lengths of the extrinsic growth front are preserved, i.e. $\forall{i} : \|\boundary[i,j]^*\| =\|\boundary[i,j]\|  $ and then use this compatible seed to construct the subsequent strip of \cell{}s. 

\textbf{Examples}
  In \figprefix{}\ref{fig:4_EXAMPLES}, we demonstrate the flexibility of our additive procedures by designing a few different classes of \hg{} using the two additive algorithms. For all designs, we also show desktop-scale karigami pieces built using laser-cut flat strips of abs plastic (\qty{1}{cm} in width) joined with eyelets at the vertex and scissor pivots (see \SIprefix{}S.V for details).

In \figprefix{}\ref{fig:4_EXAMPLES}\subfigname{A} we construct a helical karigami with loop lentgh $N=9$. 
All helical karigami are fully determined by the initial seed, which in this case is fixed to be a helix of decreasing radius (see \SIprefix{}S.IV).
In \figprefix{}\ref{fig:4_EXAMPLES}\subfigname{B}  we construct an approximation to a torus, developed from a loop-like seed and represents a more general result, namely that any axisymmetric surface (free of self crossings) and admits an karigami approximation of arbitrary accuracy (see \SIprefix{}S.IV for discussion). Here the initial seed consists of 12 equally spaced points along the nodal line of the torus, while subsequent rows are constructed by fixing an initial vertex to coincide with circles around the waist of the target torus (see \SIprefix{}S.IV for discussion).
For our third example we approximate a flat disk (\figprefix{}\ref{fig:4_EXAMPLES}\subfigname{C}).
Here we take the position of the initial seed (constrained to lie along one quarter-arc of the boundary of the target disk), leg direction \leg{2}[i,0], and boundary leg lengths of each row as variables and search for configurations in which boundary vertices lie on the circumference of the target circle. In our fourth example we map the surface of a half ellipsoid (\figprefix{}\ref{fig:4_EXAMPLES}\subfigname{D}). 
(See \SIprefix{}S.IV  for details of implementation.)


Finally, we demonstrate two structures built using the intrinsic algorithm. In \figprefix{}\ref{fig:4_EXAMPLES}\subfigname{E}, we show a surface of mixed curvature inspired by the origami eggbox pattern. 
Here regions of positive curvature are multistable, and can be ``popped through'', while collapsing the lattice removes the curvature defects and resets the system to an initial state. In \figprefix{}\ref{fig:4_EXAMPLES}\subfigname{F}, we design a surface of constant negative Gaussian curvature, using a sequence of identical sub-units which can be repeated in a loop to form one ring of the final structure (see details in \SIprefix{}S.IV).
This surface is analagous to examples demonstrated by \cite{Ono2022-nr, Ono2024-ix}.


Inspired by the simple shearing motion of a pair of scissors, we have provided a method to characterize, design and build a new linkage-based geometric metamaterial: karigami. 
The underlying geometric constraints allowed us to explore the space of collapsible karigami structures and create a selection of surfaces in 2 and 3- dimensions. While our work so far has been agnostic to the deployment pathways of these devices and to the associated mechanical behavior, these are natural questions for future investigation. 
Following in the footsteps of its fold- and cut-based cousins, we hope karigami will offer a useful geometric vocabulary for the design of mechanical devices from the micro- to architectural scales.

{\bf Ackowledgments.} We thank the Harvard NSF MRSEC 20-11754, the Simons Foundation, and the Henri Seydoux Foundation for partial financial support. 

\printbibliography

\newpage{}

\onecolumn{}

\end{document}